\documentclass{amsart}
\usepackage[utf8]{inputenc}
\usepackage{latexsym,array,delarray,amsthm,amssymb,epsfig}\usepackage{caption}
\usepackage{subcaption}
\usepackage{amscd}
\usepackage{verbatim}

\theoremstyle{plain}
\newtheorem{thm}{Theorem}[section]
\newtheorem{lemma}[thm]{Lemma}
\newtheorem{prop}[thm]{Proposition}
\newtheorem{cor}[thm]{Corollary}
\newtheorem{conj}[thm]{Conjecture}

\theoremstyle{definition}
\newtheorem{defn}[thm]{Definition}
\newtheorem{ex}[thm]{Example}

\theoremstyle{remark}
\newtheorem{rmk}{Remark}

\newcommand{\B}[1]{\mathbb #1}
\newcommand{\C}[1]{\mathcal #1}
\newcommand{\F}[1]{\mathfrak #1}

\newcommand{\ind}{\mbox{$\perp \kern-5.5pt \perp$}}

\DeclareMathOperator{\im}{Im}
\DeclareMathOperator{\rank}{rank}
\newcommand{\Span}[1]{\langle #1 \rangle}
\DeclareMathOperator{\conv}{conv}
\DeclareMathOperator{\Aut}{Aut}

\DeclareMathOperator{\JC}{JC}
\DeclareMathOperator{\CFN}{CFN}
\DeclareMathOperator{\KtP}{K3P}
\DeclareMathOperator{\KwP}{K2P}
\newcommand{\inner}[1]{\langle #1 \rangle}

\title{Dimensions of Group-based Phylogenetic Mixtures}

\date{}
\author[Ba\~{n}os]{Hector Ba\~{n}os}
\address[H.\ Ba\~{n}os]{Department of Mathematics and Statistics, University of Alaska Fairbanks, Fairbanks, AK 99775} 
\email{hdbanoscervantes@alaska.edu}

\author[Bushek]{Nathaniel Bushek}
\address[N. Bushek]{Department of Mathematics and Statistics, University of Alaska Anchorage, Anchorage, AK 99508} 
\email{nbushek@alaska.edu}

\author[Davidson]{Ruth Davidson}
\address[R. Davidson]{Department of Mathematics \& Department of Plant Biology, University of Illinois Urbana-Champaign, Urbana, IL 61801} 
\email{redavid2@illinois.edu}

\author[Gross]{Elizabeth Gross}
\address[E. Gross]{Department of Mathematics, San Jose State University, San Jose, CA 95192} 
\email{elizabeth.gross@sjsu.edu}

\author[Harris]{Pamela E. Harris}
\address[P. E. Harris]{Department of Mathematics and Statistics, Williams College, Williamstown, MA 01267} 
\email{pamela.e.harris@williams.edu}

\author[Krone]{Robert Krone}
\address[R. Krone]{Department of Mathematics and Statistics, Queens University, Kingston, ON K7L 3N6, Canada} 
\email{rk71@queensu.ca}

\author[Long]{Colby Long}
\address[C. Long]{Department of Mathematics, North Carolina State University, Raleigh, NC 27695} 
\email{celong@ncsu.edu}

\author[Stewart]{Allen Stewart}
\address[A. Stewart]{Department of Mathematics, Seattle University, Seattle, WA 98122} 
\email{stewaral@seattleu.edu}

\author[Walker]{Robert Walker}
\address[R. Walker]{Department of Mathematics, University of Michigan, Ann Arbor, MI 48109} 
\email{robmarsw@umich.edu}

\begin{document}

\maketitle
\begin{abstract}  In this paper we study group-based Markov models of evolution and their mixtures. In the algebreo-geometric setting, group-based phylogenetic tree models correspond to toric varieties, while their mixtures correspond to secant and join varieties. Determining properties of these secant and join varieties can aid both in model selection and establishing parameter identifiability.  Here we explore the first natural geometric property of these varieties: their dimension. The expected projective dimension of the join variety of a set of varieties is one more than the sum of their dimensions.  A join variety that realizes the expected dimension is nondefective.  Nondefectiveness is not only interesting from a geometric point-of-view, but has been used to establish combinatorial identifiability for several classes of phylogenetic mixture models.  In this paper, we focus on group-based models where the equivalence classes of identified parameters are orbits of a subgroup of the automorphism group of the group defining the model.  In particular, we show that, for these group-based models, the variety corresponding to the mixture of $r$ trees with $n$ leaves is nondefective when $n \geq 2r+5$.  We also give improved bounds for claw trees and give computational evidence that 2-tree and 3-tree mixtures are nondefective for small~$n$.

\end{abstract}

\section{Introduction}

A phylogenetic tree is a graphical representation of the common evolutionary history of a group of \emph{taxa}, where commonly studied taxon types include species, gene samples from microbial communities, and individuals within a single population. 
Modern gene-sequencing technology has led to a significant increase in the amount of protein, RNA, and DNA sequence data available for phylogenetic and phylogenomic inference \cite{bertelli2013rapid, gawad2016single,metzker2010sequencing}, meriting the involvement of many disciplines in the development of novel techniques for phylogenetic inference. This interaction has led to several subfields in phylogenetics, including phylogenetic algebraic geometry, which studies phylogenetic models from an algebreo-geometric framework. In this paper, we approach mixtures of group-based Markov models from this algebreo-gemetric perspective, studying the dimensions of their corresponding varieties.

The approach of studying phylogenetic and phylogenomic inference using algebraic geometry  was originally introduced to the biological community  via the concept of \emph{invariants} of tree-based Markov models \cite{cavender1987invariants,Lake1987, lake1987rate}. Since then, several invariant-based phylogeny inference approaches have been developed, and these approaches have been compared to maximum likelihood, neighbor-joining, and maximum parsimony with promising results, especially in the case when both long and short branches are present \cite{casanellas2007performance, Chifman2014,FSC2016, Rusinko2012}.

At the core of phylogenetic algebraic geometry is the fact that an algebraic variety can be associated to the set of distributions comprising a phylogenetic model \cite{ DSS2009,eriksson2005phylogenetic}, indeed, this variety is the Zariski closure of the model. Due to this correspondence, many of the properties of a phylogenetic model
can be explored with the tools of
computational algebraic geometry.
In fact, algebraic methods have proven 
useful not only for phylogenetic inference, but also for establishing identifiability \cite{allman2006identifiability, rhodes2012identifiability}, a necessary requirement for meaningful statistical inference.

Algebraic varieties can be associated not only to tree-based Markov models, but also to their mixtures. Tree-based Markov models assume that mutations occur with the same probabilities at every site along the gene sequences being studied.  Since this is an approximation to what occurs in the natural course of genetic mutation, the data is sometimes better explained by a \emph{mixture model} \cite{pagel2004phylogenetic}, in which the behavior at different sites is described by different parameter values for the same tree or even by different trees entirely.  
The variety of a mixture model is the 
\emph{join} \cite{casanellas2012algebraic} of the individual 
phylogenetic tree model varieties in 
the mixture. In general, the identifiability problem of the combinatorial parameters asks whether the set of component trees of a mixture model can be recovered from generic data. The question of identifiability has been answered for some tree-based Markov models for 2 and 3-tree mixtures \cite{Allman, long2017identifiability}. These results have relied on being able to construct polynomials that vanish on all distributions of a given mixture model as well as the dimension of the corresponding join variety.  This manuscript focuses on this second key tool in proving identifiability, the dimension of the variety.

There are several different models of sequence evolution for tree-based Markov models, each which results in a different geometry. Our study explores the  class of tree-based Markov models called 
\emph{group-based models} \cite{michalek2011geometry},
in which the transition matrices of the 
model are assumed to exhibit certain symmetries. An important observation of Evans and Speed \cite{Evans1993} is the varieties associated to group-based phylogenetic tree models are not only algebraic varieties, but toric varieties. This allows us to apply tools of computational and combinatorial algebraic geometry to their investigation.
The class of group-based models includes many 
commonly used models of sequence evolution, including the Jukes-Cantor (JC) and the Cavender-Farris-Neyman (CFN) models \cite{daskalakis2011evolutionary,hendy1994discrete,jukes1969evolution, neyman1971molecular}.

When considering a model geometrically, we first need to state its ambient space. We consider the model varieties and their joins as living in projective space. The projective dimension of a join of two projective varieties is at most the sum of their dimensions plus one.  This upper bound is typically realized, as long as it does not exceed the dimension of the ambient space, and so is referred to as the \emph{expected dimension}. 
If a join variety has the expected dimension it is said to be \emph{nondefective} and is otherwise \emph{defective}.
Proving the nondefectiveness of the join varieties associated to mixture models is the key tool in establishing identifiability results for phylogenetic mixtures.
Specifically, in
\cite{Allman, long2017identifiability}, the strategy for proving identifiability relied on showing the join varieties associated to 2 and 3-tree mixtures for the CFN, JC, and Kimura 2-parameter (K2P) models has the expected dimension for trees with few leaves.
 In this paper, we extend these results in several directions. We prove the following theorem, which not only shows that nondefectiveness holds more generally for
joins of group-based models, but also holds when there are more trees in the mixture provided the trees have a sufficient number of leaves.

\begin{thm}\label{thm: main}
Let $\mathcal{T}_1,\ldots, \mathcal{T}_r$ be 
phylogenetic $[n]$-trees with $n \geq 2r+5$, $G$ be an abelian group, and $B\subset \Aut(G)$.  Then $V^{(G,B)}_{\C T_1} * \cdots * V^{(G,B)}_{\C T_r}$ has the expected dimension.
\end{thm}

In addition to its applications to phylogenetics, Theorem \ref{thm: main} is an interesting geometric result as it adds to the growing body of knowledge on the structure of joins and secant varieties in algebraic geometry. Defectiveness of joins of toric varieties has been intensively studied in  some very specific cases, such as secants of Veronese varieties \cite{AH1995} and of Segre-Veronese varieties \cite{abo2012new,abo2013dimensions}, but little is known about  the general case or even cases outside of these examples.  

Our primary tool for proving Theorem \ref{thm: main} is a tropical version of Terracini's Lemma, a classical tool for computing the dimensions of joins.
The tropical version, introduced by Draisma, gives lower bounds for dimensions of joins of toric varieties by rephrasing questions about the dimensions of toric varieties as questions about the convex geometry of lattice points \cite{Draisma2008}.
   
This paper is organized as follows. In 
Section \ref{sec: Preliminaries}, we
introduce the group-based models and explain the Fourier transformation of \cite{Evans1993} that makes the parameterization of these models monomial and the resulting ideals toric.
We also describe toric ideals more generally and explain how Draisma's Lemma 
can be used to
establish nondefectiveness. 
In Section
\ref{sec: general group-based models}, we exploit the combinatorics of trees and use Draisma's Lemma (Theorem \ref{tropjoin}) to prove a version of the main theorem (Theorem \ref{thm: main CFN}) that holds for a class of models called the \emph{general group-based models} and for binary trees. In Section
\ref{sec: non-binary trees and other group-based models}, we prove Theorem \ref{thm: main}, which does not require that the trees in the mixture  be binary and allows for group-based models with parameter identifications. In Section \ref{sec: special cases}, we give improved bounds for some special cases and state computational results for few number of leaves. We conclude with a discussion on the applications of this work to parameter identifiability for mixture models.

\section{Preliminaries}
\label{sec: Preliminaries}

In this section, we introduce phylogenetic models and phylogenetic mixture models, and we
describe how to associate an algebraic variety to each.
We then introduce group-based phylogenetic models and a change of coordinates called the discrete Fourier transformation, in which the parameterization of a group-based phylogenetic model is monomial.
Finally, we describe the connection between group-based models 
and toric ideals and lay the groundwork for the application of Theorem \ref{tropjoin}  \cite[Corollary 2.3]{Draisma2008}. This is the primary tool that we will use to establish our main results in the subsequent sections.

\subsection{Phylogenetic Models}
\label{sec: phylogenetic models}
We follow the conventions of \cite{Evans1993} and \cite{Sturmfels2005}.
In a {\em phylogenetic model} we specify a
rooted tree $\C T'$ with $n-1$ leaves representing the 
evolutionary history of a collection of $n -1$ taxa. 
The root of $\C T'$, denoted by $\rho'$, represents the most recent common ancestor of this set of taxa.
We assume $\C T'$ has no degree-2 vertices 
other than 
the root and label the leaves by 
the set $\{1, \ldots, n-1\}$.
We then fix a finite alphabet $G = \{g_1,\ldots,g_k\}$, which in phylogenetic applications is usually chosen to be $\{A,G,C,T\}$ to represent the four DNA bases. For any choice of parameters in the phylogenetic model, we obtain a probability distribution on the set of all 
$(n - 1)$-tuples of $G$ representing the possible states at the leaves 
of $\C T'$.

To construct a distribution from a choice of parameters, let $\C V(\C T'), \C E(\C T'), \C L(\C T')$ denote the vertex, edge, and leaf vertex sets of $\C T'$ respectively. Each vertex $v \in \C V(\C T')$ has associated to it a random variable $X_v$ with state space $G$.  The distribution of states
at the root node is given by a function $\pi: G \to \B R$ with $\pi(g) = P(X_{\rho'} = g)$ for each $g \in G$. 
To each directed edge $e = (u,v)$ of $\C T'$, we associate a 
$k \times k$ stochastic \emph{transition matrix} 
$A^{(e)}$ given by 
$A^{(e)}_{ij} = P(X_v = j | X_u = i )$.  
A joint state of the random variables 
$\{X_v \ : v \in \mathcal V(\C T')  \}$ can be described by a $G$-labeling $\phi:\C V(\C T') \to G$ of the vertices.
In a phylogenetic model of DNA substitution, the labeling indicates
that at the DNA site being modeled,
the DNA base in the taxon at $v$ is $\phi(v)$. 
The probability of observing a particular labeling is then given by
\begin{equation}
    P((X_{v}= \phi(v))_{v \in \C V(\mathcal{T}')})= \pi(\phi(\rho')) \prod_{e=(u,v) \in \C E(\C T')} A^{e}_{\phi(u),\phi(v)}. 
\end{equation}

However, only the states of the random variables at the non-root leaf vertices (which represent extant species) are observable. To compute the probability of observing a particular state at the leaves of $\mathcal{T}'$, we must marginalize over all possible states of the internal vertices.  
For a $G$-labeling of the leaves $\psi:\C L(\C T') \to G$, let $p_\psi$ be the marginal over all labelings of 
$\C V(\C T')$ that extend $\psi$,
\begin{equation}
 p_{\psi} := \sum_{\phi \text{ extending } \psi}\left[\pi(\phi(\rho')) \prod_{e=(u,v) \in \C E(\C T ')} A^{e}_{\phi(u),\phi(v)}\right]. 
 \end{equation}
These $p_{\psi}$ are called {\em probability coordinates} and the entries of the transition matrices
are called the {\em stochastic parameters} of the model.  
For each choice of stochastic parameters, we obtain a probability distribution on the 
$(n-1)$-tuples of elements of $G$. Thus, the $p_\psi$ are the coordinate functions
of a
polynomial map 
$h_{\C T'}: \Theta_{\C T '} \to 
\Delta^{(k^{n-1}) -1} \subseteq
\B R^{k^{n-1}}$
from the space of stochastic parameters for $\C T'$ to the probability simplex.
We call the image of $h_{\C T'}$ the \emph{model associated to $\C T'$}, denoted $\C M_{\C T'}$.

Ignoring the stochastic restrictions on 
the parameter space, we may regard 
$h_{\C T'}$ as a complex polynomial map.
Then 
the Zariski closure $ \overline{\C M_{\C T'}} = 
V_{\C T'} \subseteq \mathbb C^{k^{n-1}}$ is 
an algebraic variety and
the set of polynomials that vanish
on this variety is the ideal
$$
I_{\C T'}
\subseteq
\mathbb{C}[p_\psi : \psi \in G^{(n-1)}].
$$
The elements of this ideal are called
{\em phylogenetic invariants} and these invariants have found many important
applications in phylogenetics \cite{casanellas2007performance, Chifman2014,FSC2016, Rusinko2012}.

\subsubsection{Phylogenetic Mixture Models}

The single tree models described above may fail to adequately describe the evolutionary history of a group of taxa for a variety of reasons. For example, due to horizontal gene transfer, hybridization, and varying rates of mutation across sites, the evolution of different sites may best be modeled by phylogenetic models with different tree parameters or by different choices of the stochastic parameters from the same phylogenetic model. Mixture models account for these phenomena by weighting the distributions from multiple models according to the proportion of sites that evolved according to each.
Thus, an $r$-tree mixture model is determined by specifying $r$ tree parameters, stochastic parameters for each tree in the model, and a mixing parameter $\omega \in \Delta^{r - 1}$ that determines the weight of each tree in the mixture. 
Thus, for an $r$-tree mixture model, we obtain a map
$h_{\C T'_1, \ldots, T'_r }: \Theta_{\C T_1 '} \times \ldots \times \Theta_{\C T_r '} \times \Delta^{r - 1} 
\to 
\Delta^{(k^{n-1}) -1} $
given by 
$$h_{\C T'_1, \ldots, T'_r }(\theta_1, \ldots, \theta_r, \omega) = \omega_1h_{\C T'_1}(\theta_1) + \ldots + \omega_rh_{\C T'_r}(\theta_r).$$ The mixture model is then denoted by $\C M_{\C T_1'} * \ldots *  M_{\C T_r'}.$

For a phylogenetic mixture model, the Zariski closure
of $ \overline{\C M_{\C T_1'} *\ldots *\C M_{\C T_r'} }$ is the variety 
$V_{\C T_1' }*\ldots*V_{\C T_r' } \subseteq \mathbb C^{k^{n-1}}$ which is the 
\emph{join} of the varieties associated to each tree in the mixture. 
We formally define this term and explore the connection between mixture models
and join varieties in Section \ref{sec: Toric varieties}.
Our goal in this paper, will be to prove results for the dimensions of the join
varieties associated to a particular class of phylogenetic models, called the
group-based models.

\subsection{Group-based Phylogenetic Models}
\label{sec: phylogenetic models}

In a {\em group-based model}, the alphabet $G$ is given the additional structure of an abelian group.  
\begin{defn}\label{def:phylogeneticmodel}
A phylogenetic model is group-based if
for each edge 
$e \in \mathcal{E}(\mathcal{T}')$
there exists a transition function 
$f^{(e)}:G \to \B R$ such that for all 
$1 \leq i, j \leq k$,
$A^{(e)}_{ij} = f^{(e)}(g_i - g_j)$.
\end{defn}

\noindent
For example, the {\em Cavender-Farris-Neyman model} and the {\em Kimura 3-parameter model} are both group-based models.  In the  Cavender-Farris-Neyman model (CFN), $G = \{0,1\}$ is given the group structure $\B Z/2\B Z$.  Each transition matrix is specified by 2 parameters, one for each element of $G$,
 \[ A^{(e)} = \begin{pmatrix}
     \alpha & \beta  \\
     \beta  & \alpha
    \end{pmatrix}. \]
In the  Kimura 3-parameter model (K3P), $G = \{A,G,C,T\}$ is given the group structure 
of $(\B Z/2\B Z) \times (\B Z/2\B Z)$ with $A$ defined to be the identity element.  The transition matrices of this model have the form
 \[ A^{(e)} = \begin{pmatrix}
     \alpha & \beta  & \gamma & \delta \\
     \beta  & \alpha & \delta & \gamma \\
     \gamma & \delta & \alpha & \beta  \\
     \delta & \gamma & \beta  & \alpha
    \end{pmatrix}. \]

For the group-based models, it will be convenient to modify the tree parameter by adding a leaf to the root node of $\C T'$. Call the new tree with $n$ leaves $\C T$ and denote the new leaf vertex by $\rho$.  Orient the edges
of $\mathcal{T}$ away from the leaf labeled $\rho$. Define $X_\rho$ to be the random variable with state space $G$ where $P(X_\rho(g)) = 1$ if $g$ is the identity and zero otherwise. Let the transition function
on the edge $(\rho,\rho')$ be defined by $f^{(\rho,\rho')}(g) = \pi_g$ for all $g \in G$.  Therefore,
the distribution of $X_{\rho'}$ remains $\pi$, but we have removed the special distinction of the 
root distribution so that now all parameters are encoded by the set of functions $\{f^{(e)}\}_{e \in \C E(\C T)}$. 
We now think of the 
tree as being ``rooted" at $\rho$, though in fact the 
tree parameter of the model is now an unrooted tree with no degree two vertices. 
We now rewrite the map given in (2), so that for a $G$-labeling of the leaves $L(\C T)\setminus \{\rho\}$, 
 \[ p_{\psi} := \sum_{\phi \text{ extending } \psi}\left[ \prod_{(u,v) \in \C E(\C T)} f^{((u,v))}(\phi(u) - \phi(v))\right]. \]
Notice that the stochastic parameters of a group-based model are the values of the transition functions $f^{(e)}$.
Now for the unrooted tree $\C T$ we have a map $h_{\C T}: \Theta_{\C T } \to 
\Delta^{(k^{n-1}) -1} \subseteq
\B R^{k^{n-1}}$ and an associated algebraic variety $V_{\C T}$.

If we do not place any additional
restrictions on the functions 
$\{f^{(e)}\}_{e \in \C E(\C T)}$, other
than that they give a probability distribution, 
then
the model associated to $G$ is called the \emph{general group-based model} associated to $G$. 
For example, the CFN model
described above is the general group-based model associated to $\mathbb{Z}/2\mathbb{Z}$ and the K3P model is the general group-based model
associated to $\mathbb{Z}/2\mathbb{Z} \times \mathbb{Z}/2\mathbb{Z}$. 
However, in some models, the parameters associated to some group elements may be identified. For example, both the {\em Kimura 2-parameter model} (K2P) and the {\em Jukes-Cantor model} (JC) can be obtained from K3P by identifying parameters. 
In the K2P model, the parameters for $C$ and $T$ are identified, while in the JC model the parameters for $G,C,$ and $T$ are identified.  Thus the transition matrices for the K2P and JC models have the respective forms
 \[ A^{(e)} = \begin{pmatrix}
     \alpha & \beta  & \gamma & \gamma \\
     \beta  & \alpha & \gamma & \gamma \\
     \gamma & \gamma & \alpha & \beta  \\
     \gamma & \gamma & \beta  & \alpha
    \end{pmatrix}, \quad\quad 
    A^{(e)} = \begin{pmatrix}
     \alpha & \beta  & \beta  & \beta  \\
     \beta  & \alpha & \beta  & \beta  \\
     \beta  & \beta  & \alpha & \beta  \\
     \beta  & \beta  & \beta  & \alpha
    \end{pmatrix}. \]

The identification of parameters can be specified by an equivalence relation on $G$.  We require the equivalence classes of this relation to be the orbits of some subgroup of $\Aut(G)$.  In particular this means that the identity is always in its own class.
Therefore a group-based phylogenetic tree model is specified by the data of a finite group $G$, an 
$n$-leaf directed tree $\C T$, and a subgroup $B$ of $\Aut(G)$.  
Therefore, we will now use the notation
$\C M_{\C T}^{(G,B)}$ and $V_{\C T}^{(G,B)}$ for the model and variety of the group-based model $(G,B)$ on $\C T$. 
The specific pairs $M = (G,B)$ described above are 
\begin{itemize}
 \item CFN $= (\B Z/2\B Z, \{1\})$,
 \item JC  $= (\B Z/2\B Z \times \B Z/2\B Z , \F S_3)$,
 \item K2P $= (\B Z/2\B Z \times \B Z/2\B Z , \F S_2)$,
 \item K3P $= (\B Z/2\B Z \times \B Z/2\B Z , \{1\})$,
\end{itemize}
where $\Aut(\B Z/2\B Z \times \B Z/2\B Z )$ is identified with the permutation group 
$\F S_3$. 
Observe also that the 
general group-based models are then precisely those models for
which $B = \{1\}$. 
Our strategy for proving the main result will first be to prove some results for general group-based models in Section \ref{sec: general group-based models} and then to show in 
Section \ref{sec: non-binary trees and other group-based models}
that they still hold for models in which we identify certain parameters.

\subsection{The Fourier Transformation}

In this section, we describe the parameterization of a group-based 
phylogenetic model
in the Fourier coordinates.
In these coordinates, the parameterization
is seen to be monomial and consequently
the varieties associated to these models are toric. 
Consider a group-based model
specified by $(G,B)$
on the $n$-leaf tree $\mathcal{T}$.
Suppose that $\C T$ has $m$ edges and that
there are $l+1$ orbits of $B$.  
Then the model has a total of $m(l+1)$ parameters, but only $ml$ are independent because $\sum_{g \in G} f^{(e)}(g) = 1$ for each $e \in \C E(\C T)$.
The observation of Evans and Speed \cite{Evans1993} is that there is a linear change of coordinates in which $h_{\C T}$ becomes a monomial map.  This implies that, in the new coordinates, the image of $h_{\C T}$ is a toric variety, which aids in the search of phylogenetic invariants.

Let $\hat{G}$ denote the character group of $G$, consisting of all group homomorphisms $\chi:G \to \B C^\times$.  Note that since $G$ is abelian $\hat{G}$ is isomorphic to $G$ itself.  Define $\inner{\chi,g} := \chi(g)$.  For the transition function $f^{(e)}:G \to \B C$, the Fourier transform $\hat{f}^{(e)}:\hat{G} \to \B C$ is given by
 \[ \hat{f}^{(e)}(\chi) = \sum_{g \in G} \inner{\chi,g} f^{(e)}(g). \]
 
\noindent Similarly we can define a Fourier transform of the probability coordinates.  Let $\xi'$ denote a $\hat{G}$-labeling of $\C L(\C T)\setminus \{ \rho\}$.  
For each such $\xi'$ let
 \[ \hat{p}(\xi') := \sum_{\psi} \prod_{v \in \C L(\C T)\setminus \rho} \inner{\xi'(v),\psi(v)} p_\psi. \]
 
 \noindent As we see in the following theorem, the transformed probability coordinates can be written in terms of the transformed transition functions.
 
 \begin{thm}\cite{Evans1993}
 \label{thm: EvansSpeed}
Let $\C L(e)$ denote the set of leaves on the arrow-side of edge $e$ (i.e. its descendants).  Then
  \[ \hat{p}(\xi') = \prod_{e \in \C E(\C T)} \hat{f}^{(e)}\bigg(\prod_{v \in \C L(e)} \xi'(v)\bigg). \]
\end{thm}

We call the new coordinates, $\hat{p}(\xi')$ the 
{\em Fourier coordinates} and the values of the transformed transition functions, $\hat{f}^{(e)}(\chi)$, the {\em Fourier parameters}. 
For the rest of this paper, the notation 
$V^{M}_{\C T}$ denotes
the model variety in the space of Fourier coordinates.  
In the case that $M = (G,\{1\})$ 
we may simply write $V^G_{\C T}$.

We would like to remove the asymmetry in this description caused by the root $\rho$.  Each $\hat{G}$-labeling $\xi'$ of $\C L(\C T)\setminus \{ \rho\}$ can be uniquely extended to a labeling $\xi$ of $\C L(\C T)$ by assigning
 \[ \xi(\rho) = 
  \prod_{v \in \C L(\C T)\setminus \{ \rho \} } \xi'(v)^{-1}. 
 \]
Such labelings are called {\em consistent leaf labelings} of $\C T$. Notice that the consistent leaf labelings are exactly the $k^{n-1}$ leaf labelings for which the product of all $n$ leaf labels is equal to the identity.

We label the Fourier coordinates by consistent leaf labelings as
 $q_{\xi} := \hat{p}(\xi').$
For each consistent leaf labeling $\xi$, there is an associated {\em consistent edge labeling} $\tilde{\xi}:\C E(\C T) \to \hat{G}$ given by
 \[ \tilde{\xi}(e) := \prod_{v \in \C L(e)} \xi(v) \]
where $\mathcal L(e)$ is defined as in Theorem \ref{thm: EvansSpeed}. Consistent edge labelings are characterized by the property that for each internal vertex $v$,
 \[ \prod_{e \in v_{\mathrm{in}}} \tilde{\xi}(e) \prod_{e \in v_{\mathrm{out}}} \tilde{\xi}(e)^{-1} = 1 \]
where $v_{\mathrm{in}}$ and $v_{\mathrm{out}}$ denote the incident incoming and outgoing edges to $v$ respectively.  Finally this lets us describe the monomial map $\hat{h}_{\C T}$ from the space of Fourier parameters to the Fourier coordinates by the equation
 \[ q_\xi = \prod_{e \in \C E(\C T)} \hat{f}^{(e)}(\tilde{\xi}(e)). \]
 
 \begin{rmk}
  This description of $V_{\C T}^M$ does not depend on the orientation of $\C T$.  To see this, note that for any fixed consistent leaf labeling $\xi$, reversing the orientation of an edge $e$ inverts $\tilde{\xi}(e)$, but does not change the labels of the other edges.  
For any choice of parameters, replacing each $\hat{f}^{(e)}(\chi)$ with $\hat{f}^{(e)}(\chi^{-1})$ produces the same point in Fourier probability space.
Thus, we can remove any special distinction of the root leaf $\rho$ by choosing an
arbitrary orientation of the edges.
Moreover, if $G$ has characteristic 2, such as $\B Z/2\B Z$ or $(\B Z/2\B Z) \times (\B Z/2\B Z)$, then the map $\hat{h}$ is itself invariant under changes of orientation of $\C T$, and so orientation can be ignored entirely.
\end{rmk}

It is common in phylogenetic applications to assume that the root distribution is uniform. When this is the case, the construction described above differs slightly. In particular, we now have
$\hat{f}^{(\rho,\rho')}(\chi) = 1/|G|$ 
if 
$\chi$ is the identity and zero otherwise
\cite{Sturmfels2005}. 
Thus, the only coordinates we need
to consider are ones corresponding to 
consistent leaf labelings that satisfy
 \[\prod_{v \in \C L(e)\setminus \{ \rho\} } \xi(v) = 1. \]
These are exactly the consistent leaf labelings for the model on $\C S$ where $\C S$ is the unrooted tree obtained by removing
the root of $\C T$ and suppressing the resulting degree two vertex---to define the model on $\C S$, we can regard any of the $n-1$ leaves as the root leaf. Consequently, if we assume the root distribution is uniform, 
many of the coordinates for the model variety of $\mathcal{T}$ are zero, and we can regard this variety as the model
variety $V_{\C T'}^M$ where the root distribution is arbitrary.  Therefore, when we assume the root distribution is uniform, we interpret the model variety for an $n$-leaf unrooted tree as corresponding to a statistical model for $n$ taxa. Thus, all the results that we prove in this paper for group-based model varieties still apply when the root distribution is assumed to be uniform.

When $B$ is trivial (so all elements of $G$ receive distinct parameters), $V^{(G,\{1\})}_{\C T}$ is the image of $\hat{h}_{\C T}:\B C^{mk} \to \B C^{k^{n-1}}$, the parameterization map in the Fourier coordinates.
For $B$ non-trivial, the identification of stochastic parameters induces an identification of Fourier parameters.
An automorphism $\alpha$ of $G$ induces an automorphism $\alpha^*$ of $\hat{G}$ by defining
 \[ \inner{\alpha^*(\chi), g} = \inner{\chi, \alpha(g)}. \]
Therefore $B \subseteq \Aut(G)$ has a corresponding subgroup $\hat{B} \subseteq \Aut(\hat{G})$.  If we insist that for any $\alpha \in B$, $f^{(e)}(\alpha g) = f^{(e)}(g)$, then it can be shown that
 \[ \hat{f}^{(e)}(\alpha^* \chi) = \hat{f}^{(e)}(\chi). \]
 That is, the orbits of $B$ are mapped into
 orbits of $\hat B$ under the Fourier transform. Therefore, if two stochastic
 parameters are assumed to be equal, they are mapped
 to two identical Fourier parameters.

 As a result, there are only $l+1$ distinct Fourier parameters for each edge and 
$V_{\C T}^{(G,B)}$ is the image
of the monomial map
$\hat{h}_{\C T}: \B C^{m(l+1)} \to \B C^{k^{n-1}}$.
The fact that 
$ \hat{f}^{(e)}(\alpha^* \chi) = \hat{f}^{(e)}(\chi)$ is the reason that we insist that the equivalence classes of probability parameters are orbits of $B$.
Otherwise, it is possible to have identified probability parameters
mapping to distinct Fourier parameters
\cite[Appendix A]{Michalek2011}.
In such a case, the Fourier parameterization is monomial, but the Fourier parameters are not algebraically independent and so the resulting ideals are not toric.

\begin{rmk}\label{rmk: dual group}
Since $G$ is a finite abelian group, $G$ and $\hat{G}$ are isomorphic and their elements can be identified. 
None of our results depend on the particular identification used and so from here on we will not carefully distinguish between the two. We will label the Fourier coordinates by consistent leaf labelings using elements of $G$ and also label Fourier parameters by elements of $G$.
\end{rmk}

\begin{ex}

This example demonstrates the parameterization of one of the Fourier coordinates for the K3P model on the  4-leaf unrooted tree 
$\C T$ pictured below.
The tree $\C T$ is constructed by attaching a leaf to the root of a 3-leaf rooted tree.
For this model,  
$G = 
\mathbb{Z}/2\mathbb{Z} \times 
\mathbb{Z}/2\mathbb{Z}$
and $B = \{1\}$. As noted in Remark 1, because $G = 
\mathbb{Z}/2\mathbb{Z} \times 
\mathbb{Z}/2\mathbb{Z}$, the parameterization will remain unchanged if we reorient edges in $\C T$.

\begin{center}
    \includegraphics[width=4cm]{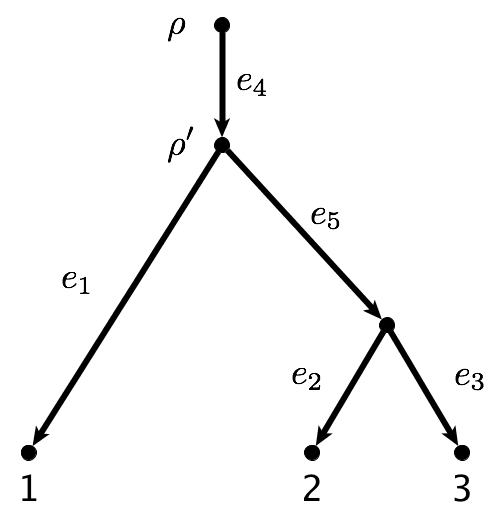}
\end{center}

The leaf
labeling
$\hat \xi
=((0,1), 
(1,0),
(1,1),
(0,0))$,
where the last
element in the index is the label
of the ``root" leaf, is consistent since 
$(0,1) + 
(1,0) +
(1,1) +
(0,0) = 
(0,0)$.
Using $a^i_{g}$
for the Fourier parameter associated to the group element $g$ on the edge $e_i$, 

$$q_{
((0,1), 
(1,0),
(0,0),
(1,1))} = 
a^1_{(0,1)}
a^2_{(1,0)}
a^3_{(1,1)}
a^4_{(0,0)}
a^5_{(1,0) + (1,1)}
=
a^1_{(0,1)}
a^2_{(1,0)}
a^3_{(1,1)}
a^4_{(0,0)}
a^5_{(0,1)}.$$

\end{ex}

The reason for introducing the Fourier parameterization is so that we may work with the toric
ideals $V_{\C T}^{(G,B)}$. Because of the linearity of the transform, the parameterization
map for an $r$-tree mixture in the Fourier coordinates is given by
$$\hat h_{\C T'_1, \ldots, T'_r }(\hat \theta_1, \ldots, \hat \theta_r, \omega) = \omega_1\hat h_{\C T'_1}(\hat \theta_1) + \ldots + \omega_r\hat h_{\C T'_r}(\hat \theta_r).$$ 
Thus, in the Fourier coordinates, the variety for a mixture model is the join of the varieties for each
tree in the mixture in the Fourier coordinates. Therefore, to determine the dimension of a phylogenetic
mixture model, we can utilize the results we describe in the next section about the dimensions of joins
of toric varieties.

\subsection{Toric varieties}
\label{sec: Toric varieties}
Because $V_{\C T}^M$ is the image of a monomial map $\hat{h}_{\C T}:\B C^{m(l+1)} \to \B C^{k^{n-1}}$, it is a complex toric variety.  Since the map is homogeneous we can consider $V_{\C T}^M$ as a projective variety in $\B P^{k^{n-1}-1}$.  The map $\hat{h}$ can be described by a $m(l+1) \times k^{n-1}$ matrix $A$. The columns of which are the exponent vectors of the monomials parameterizing each Fourier coordinate.  Let $\C A \subseteq \B R^{m(l+1)}$ denote the set of column vectors of $A$ and let $P \subseteq \B R^{m(l+1)}$ be the convex hull of $\C A$.  The geometry of $P$ and $V_{\C T}^M$ are closely tied in many ways. We will be primarily interested in dimension,
 \[ \dim V_{\C T}^M = \dim_{\B R} P = \rank A - 1 \]
where $\dim V_{\C T}^M$ denotes the projective dimension \cite{sturmfels1996grobner}.

Label the equivalence classes of $G$ induced by $B$ by the integers from $0,\ldots,l$, with 0 labeling the equivalence class containing only the identity.
Fixing $B$, let $\delta:G \to \B R^{l+1}$ be the map sending each group element to $e_i$ where $i$ is the label of its equivalence class and $e_0,\ldots, e_l$ are the $l+1$ standard basis vectors.  The image of $\delta$ is the set of vertices of a standard $l$-simplex $\Delta_l := \conv(e_0,\ldots,e_l)$, and its affine span is the hyperplane $K$ defined by $x_0 + \cdots + x_l = 1$ where $x_i$ denotes the $i$th coordinate.

Each of the vectors in $A$ comes from a consistent edge-labeling of $\C T$, $(g_1,\ldots,g_m) \in G^m$, by the map
 \[ \delta^m: G^m \to (\B R^{l+1})^m\]
 \[ (g_1,\ldots,g_m) \mapsto (\delta(g_1),\ldots,\delta(g_m)). \]
The image of $\delta^m$ is the set of lattice points corresponding to all possible $G$ edge-labelings of the graph (not just consistent ones), so $\C A \subseteq \im \delta^m$.  The convex hull of $\im \delta^m$ is the polytope $\Delta_l^m$, e.g. for the CFN model, $\conv(\im \delta^m) = \Delta_1^m$ is an $m$-dimensional cube. The polytope $\Delta_l^m$ has dimension $lm$ and its affine span is $K^m$.  The polytope $P^{M}_{\C T}$ associated to $V_{\C T}^M$ is contained in  $\Delta_l^m$, and thus,
 \[ \dim V_{\C T}^M \leq lm. \]

\begin{rmk}
The map $\hat{h}_{\C T}$ is multi-homogeneous, which is reflected in the fact that $P^{M}_{\C T}$ is contained in the codimension-$m$ space $K^m \subseteq \B R^{(l+1)m}$.  It is sometimes convenient to consider the dehomogenized map by projecting away coordinate $x_0$ for each edge. The projected polytope $P^{M}_{\C T} \subseteq \B R^{lm}$ differs only by a linear change of coordinates from the above, but now affinely spans the ambient space in the case that it realizes the lower bound and has dimension $lm$. 
\end{rmk}

Given varieties $W_1,\ldots, W_r \subseteq \B P^{N-1}$, we denote their join by $W_1 * \ldots *W_r$, which is defined as the Zariski closure of the set of linear spaces defined by one point from each variety.  The variety $W_1 *\ldots* W_r$ can be considered as the closure of the image of the map 
 \[ W_1 \times\ldots \times W_r \times \B P^{r-1} \to \B P^{N-1}, \]
 \[ (p_1,\ldots,p_r,[c_1:\ldots: c_r]) \mapsto c_1p_1 + \ldots + c_rp_r. \]
From this map and the discussion at the end of Section \ref{sec: general group-based models}, it is evident that the variety associated to the mixture model $M$ on trees $\mathcal{T}_1,\ldots,\mathcal{T}_r$
 is the join variety
$\mathcal{V}_{\C T_1}^{M}*\ldots* 
\mathcal{V}_{\C T_r}^{M}$.
Thus, to establish our main result for 
phylogenetic mixture models we will utilize
techniques for bounds on the dimension
of join varieties.

Just by counting parameters, we have the following upper bound on the dimension of $W_1 * \ldots * W_r$,
 \[ \dim(W_1 * \ldots * W_r) \leq \dim W_1 + \ldots + \dim W_r + (r -1). \]
Another upper bound on $\dim(W_1 * \ldots * W_r)$ is the dimension of the ambient space, $N-1$.  The {\em expected dimension} of $W_1 * \ldots*W_r$ is
 \[ \min\{\dim W_1 + \ldots + \dim W_r + (r -1), N-1\}. \]
If the dimension is less than the expected dimension, $W_1 * \ldots * W_r$ is said to be {\em defective}.

For $W \subseteq \B P^{N-1}$, the join $W * W$ is called the {\em secant} of $W$ (or more specifically the {\em second secant} of $W$), also written $\sigma(W)$ or $\sigma_2(W)$.  For any integer $r \geq 1$, the {\em $r$th secant} of $W$ is
 \[ \sigma_r(W) := \underbrace{W * \cdots * W}_{r \text{ times}}. \]
The expected dimension of $\sigma_r(W)$ is
 \[ \min\{r(\dim W + 1) - 1, N\}. \]
 
To prove the main theorem, we will rely heavily on a tool developed by Draisma. This tool allows us to determine lower bounds on the dimensions of joins and secants of toric varieties using tropical geometry \cite{Draisma2008} (Theorem \ref{tropjoin}).  Let $W_1,\ldots,W_r$ be projective toric varieties in $\B P^{N-1}$, and let $A_i$ be the $m_i \times N$ matrix associated to $W_i$ for $i = 1,\ldots,r$.  Let $v = (v_1,\ldots,v_r)$ be a sequence of linear functionals, with each $v_i:\B R^{m_i} \to \B R$, considered as a row vector.  
Then $v_i A_i$ is a row vector in $(\B R^N)^*$.
Let $W_i(v) \subseteq [N]$ denote the set of positions $j$ such that the $j$th entry of $v_i A_i$ is strictly less than the $j$th entry of $v_l A_l$ for all $l \neq i$.  Thus $W_1(v),\ldots,W_r(v)$ are disjoint, and for a generic choice of $v$ form a partition of $[N]$.

Now let $D_i(v) \subseteq \B R^{m_i}$ be the set of column vectors of $A_i$ in the positions given by $W_i(v)$.  Let $\rank D_i(v)$ denote the dimension of the span of $D_i(v)$.  This is one more than $\dim \conv(D_i(v))$, the dimension of the affine span of $D_i(v)$.

\begin{thm}[Corollary 2.3 of \cite{Draisma2008}]\label{tropjoin}
For any choice of $v \in \prod_{i=1}^r (\B R^{m_i})^*$,
 \[ \dim(W_1 * \cdots * W_r) \geq \rank D_1(v) + \cdots + \rank D_r(v)-1. \]
\end{thm}
\noindent
It follows that if there exists $v$ such that $\rank D_i(v) = m_i$ for all $i = 1,\ldots,r$, then $W_1 * \cdots * W_r$ has the expected dimension.

In the case of a secant, the picture is a bit simpler.  Let the toric variety $W \subseteq \B P^{N-1}$ have $m \times N$ matrix $A$ and let $\C A \subseteq \B R^m$ denote the set of column vectors of $A$.  Let $v = (v_1,\ldots,v_r)$ be a sequence of linear functionals on $\B R^m$.  
The sequence $v$ divides $\B R^m$ into open convex regions $R_1(v),\ldots,R_r(v)$ defined by
 \[ R_i(v) := \{ p \in \B R^m \mid v_i(p) < v_l(p) \text{ for all } l \neq i \} \]
for $i = 1,\ldots,r$.
\begin{cor}\label{tropsec}
For any choice of $v \in \prod_{i=1}^r (\B R^m)^*$,
 \[ \dim \sigma_r(W) \geq \rank(\C A \cap R_1(v)) + \cdots + \rank(\C A \cap R_r(v)) -1. \]
\end{cor}
In the particular case of $r = 2$, the regions $R_1(v)$ and $R_2(v)$ can be described as the open half-spaces on either side of a hyperplane $\C H$.  In this case we denote the regions by $\C H^+$ and $\C H^-$.
For larger $r$, it may also be useful to partition the space by $r-1$ hyperplanes (although these are not the only sort of partitions allowed).

\begin{prop}\label{prop:hyperplanes}
Let $\C H_1,\ldots,\C H_{r-1}$ be hyperplanes through $P$, with no two intersecting in $P$.  Let $P_1,\ldots,P_r$ be the connected components of $P \setminus(\C H_1 \cup \cdots \cup \C H_{r-1})$.
Then there is a sequence of functionals $v = (v_1,\ldots,v_r)$ such that $P_i = R_i(v) \cap P$.
\end{prop}

\begin{proof}
 Because the hyperplanes do not intersect in $P$, for each $i \neq j$, $\C H_j \cap P$ is contained in $\C H_i^+$ or $\C H_i^-$.  Therefore by reindexing, and choosing plus and minus labels appropriately, we have
 \[ (\C H_1^+ \cap P) \supseteq (\C H_2^+ \cap P) \supseteq \cdots \supseteq (\C H_{r-1}^+ \cap P). \]
 The connected components of $P \setminus(\C H_1 \cup \cdots \cup \C H_{r-1})$ are then $P_1 = \C H_1^- \cap P$, $P_i = \C H_{i-1}^+ \cap \C H_{i}^- \cap P$ for $1 < i < r$, and $P_r = \C H_{r-1}^+ \cap P$.
 For each $\C H_i$ choose a functional $\ell_i$ which vanishes on $\C H_i$ and is positive on $\C H_i^+$.  Then let $v_i = -\ell_1 - \cdots - \ell_{i-1}$ for $i = 1,\ldots,r$.  For any $p \in P_i$, $\ell_j(p) > 0$ for $j < i$ and $\ell_j(p) < 0$ for $j \geq i$.  Therefore $v_i(p)$ is the unique minimum among $v_1(p),\ldots,v_r(p)$, so $p \in R_i(v)$.
\end{proof}

\section{Dimension of joins of group-based models}
\label{sec: general group-based models}

In this section we prove two intermediary theorems on our way to proving Theorem \ref{thm: main}.
Together, these two theorems establish the main result for all groups so long as the trees $\C T_1,\ldots,\C T_r$ are binary and $B$ is trivial.
We require two theorems, 
since slightly different arguments are needed when 
$G = (\mathbb{Z}/2\mathbb{Z})$.   This group-based model is particularly relevant to phylogenetic applications as it is exactly the CFN model discussed in the introduction.  We handle that case first, and then the case $|G| > 2$.  In the next section we will generalize this result for arbitrary trees $\C T_1,\ldots,\C T_r$ and for $B$ an arbitrary subgroup of $\Aut(G)$.

\subsection{Joins for the CFN model}

We will prove the following theorem using Theorem \ref{tropjoin} (Draisma's Lemma). In this section, we will use 
the notation $V^{\B Z/2\B Z}_{\C T}$ in place of $V^{(\B Z/2\B Z,\{1\})}_{\C T}$.

\begin{thm}\label{thm: main CFN}
Let $\mathcal{T}_1,\ldots, \mathcal{T}_r$ be 
binary phylogenetic $[n]$-trees with $n \geq 2r+5$.  Then $V^{\B Z/2\B Z}_{\C T_1} * \cdots * V^{\B Z/2\B Z}_{\C T_r}$ has the expected projective dimension, $r(2n-3)+ r -1$.
\end{thm}

Each toric variety $V^{\B Z/2\B Z}_{\C T_i}$ is parametrized by a monomial map with set of exponent vectors $\C A_i \in \B R^{2(2n-3)}$ which correspond to the constistent labeleings of $\C T_i$.  The affine span of $\C A_i$ has dimension $2n-3$ thus the projective dimension of $V^{\B Z/2\B Z}_{\C T_i}$ is $2n-3$.

Per Draisma's Lemma, we demonstrate the existence of a set of functionals $v = (v_1,\ldots,v_r)$ that partition the consistent labelings into sets $W_1(v),\ldots,W_r(v)$.  Then  for all $1 \leq i \leq r$, $D_i(v)$ is a proper subset of $\C A_i$ constructed from $W_i(v)$ according to Theorem \ref{tropjoin}. The goal is to show that the dimension of the affine span of $D_i(v)$ is the same as the dimension of the affine span of $\C A_i$, namely $2n-3$.

The general strategy will be to partition the consistent labelings based on the number of leaves not labeled by the identity.  The subset of $\C A_i$ of vectors corresponding to labelings with exactly $c$ non-identity leaf labels has affine span of dimension at most $2n-4$ because this imposes one linear constraint on the vectors. In the following lemma we show that for even $2 \leq c \leq n-5$, the affine span has exactly that dimension.  (Note that the rank of the linear span is one larger than the affine dimension.)
In what follows, we use the notation
$x^\epsilon_g$ to denote the coordinate that
corresponds to the Fourier parameter
associated to $g$ on the edge $\epsilon$. 
We also slightly abuse notation and interpret 
$\C L (\C T)$ as the set of leaf
edges of $\C T$ or as the set of leaf vertices of $\C T$ depending on the context.  Similarly, we not carefully distinguish between a leaf vertex and the leaf edge leading to that vertex.
An example of the construction
from this lemma is illustrated in Example
\ref{ex: Z2 Tree}.

\begin{lemma}
\label{lem: CFN slice}
Let $G = \B Z/2\B Z $ and let $\C T$ be a binary phylogenetic $[n]$-tree. 
Let $\C A$ be the set of vectors describing the monomial map 
parameterizing $V_{\mathcal{T}}^{\B Z/2\B Z}$. Let $\C S_c \subseteq \mathbb R^{2(2n-3)}$ be the hyperplane defined by
\[ \sum_{\ell \in \C L(\C T)} x^\ell_1 = c. \]
For even $2 \leq c \leq n-5$, $\rank\langle \C A \cap \C S_c \rangle= 2n-3$.
\end{lemma}

\begin{proof}

Let $\pi: \B R^{2(2n-3)} \to \B R^{2n-3}$ be the projection that forgets coordinates $x_0^{\epsilon}$ for identity element $0 \in \B Z/2\B Z$ and each edge $\epsilon \in \C E(\C T)$.  Then $\pi(\C A \cap \C S_c)$ is a set of $0/1$ vectors with $c$ leaf coordinates equal to 1.

Let $E$ be an internal edge of $\mathcal{T}$.
Redraw $\mathcal{T}$ as below where each $r_j$ is a
rooted subtree of $\mathcal{T}$ with root $\rho_j$.

\begin{center}
\includegraphics[width=5cm]{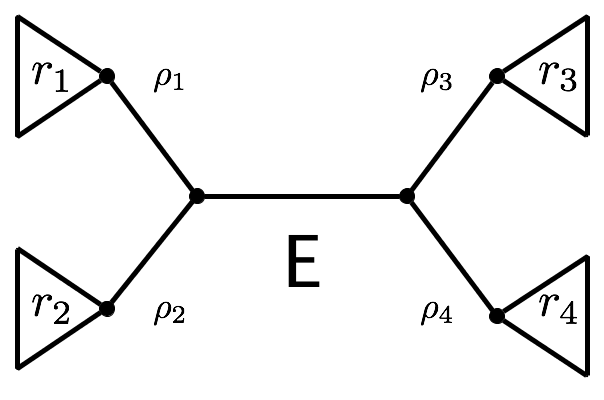}
\end{center}
For $1 \leq j \leq 4$, choose $L_j$ to be a subset of the leaves of $r_j$  so that each
$|L_j|$ is odd and so that
$|L_1| + |L_2| + |L_3| + |L_4| =  c$.
This is always possible since $c \leq n-5$.
Now, label all the leaves in 
$L_1 \cup L_2 \cup L_3 \cup L_4$ by $1$ to give a consistent
leaf-labeling of $\mathcal{T}$. This induces a consistent edge-labeling of $\mathcal{T}$. Observe that in each rooted subtree $r_j$ there is a unique leaf $\lambda_j$ such that the path from $\rho_j$ to $\lambda_j$
involves only edges labeled by $1$. 

Let $F_{kl} \in \pi(\C A \cap \C S_c)$ be the vector corresponding to the subforest of $\mathcal{T}$ induced by labeling all the leaves in 
$$\left ( \bigcup_{1 \leq j \leq 4} (L_j\setminus \{\lambda_j\}) \right )  \cup \{\lambda_k, \lambda_l \}$$ by $1$ and all other leaves by 0. Let $F_{E}$ be the vector that corresponds to the subforest
induced by labeling all of the leaves in 
$$ \bigcup_{1 \leq j \leq 4} (L_j\setminus \{\lambda_j\})$$
by $1$ and all other leaves by 0. The key observation
is that 
$$F_{kl} = F_E + \displaystyle \sum_{\epsilon \in p(\lambda_k,\lambda_l)}
e^\epsilon_1,$$ where $p(\lambda_k,\lambda_l)$ is the path between $\lambda_k$ and $\lambda_l$.
Therefore, we have
$$F_{13} + F_{24}
 - F_{12}  - F_{34} = 2e_1^{E}.$$
 
Now suppose that $\ell$ is a leaf edge of $\mathcal{T}$. Let $\Lambda$ be a $c$-element subset of the leaves and label each leaf in this subset by 1 and all of the rest by 0. Since $c$ is even, this is a consistent labeling. The vector corresponding to the subforest induced by this labeling is in $\pi(\C A \cap \C S_c)$. Moreover, 
since we have already shown that $\pi(\C A \cap \C S_c)$ contains
$e^\epsilon_1$ for any internal edge $\epsilon$, the vector 
$F_\Lambda = \sum_{\epsilon \in \Lambda}
e^\epsilon_1$ must also be in 
$\pi(\C A \cap \C S_c)$. For each leaf $\epsilon \neq \ell$, let $\Lambda_\epsilon$
be any $c$-element subset of the leaves that
contains $\epsilon$ but not $\ell$. Then
$$F_{(\Lambda_\epsilon \setminus \{\epsilon\} )\cup \{\ell\}} - F_{\Lambda_\epsilon} = e^\ell_1 - e^\epsilon_1,$$
and, we get,  $$F_{\Lambda} + \sum_{\epsilon \in \Lambda} (e_1^{\ell} - e_1^{\epsilon}) = ce^\ell_1.
$$ 
Since we can repeat this procedure for every leaf edge $\ell$ and we have shown that $e_1^{E} \in \langle \C A \cap \C S_c \rangle$ for every internal edge $E$, we can conclude that 
$\rank \langle \C A \cap \C S_c \rangle \geq \rank \langle \pi(\C A \cap \C S_c) \rangle \geq 2n -3$. 
\end{proof}

\begin{ex}
\label{ex: Z2 Tree}

The figure below represents a consistent leaf-labeling 
for which $c = 16$. Both the blue and red vertices are elements of the $L_j$ ($|L_1| = 7$ and 
 $|L_2| =  |L_3| =  |L_4| = 3$) and the blue vertices are the $\lambda_i$. All of the colored edges correspond to non-zero entries in the vector $F_{12}$ and the red colored edges to the non-zero entries in $F_E$.

\begin{center}
\includegraphics[width=5cm]{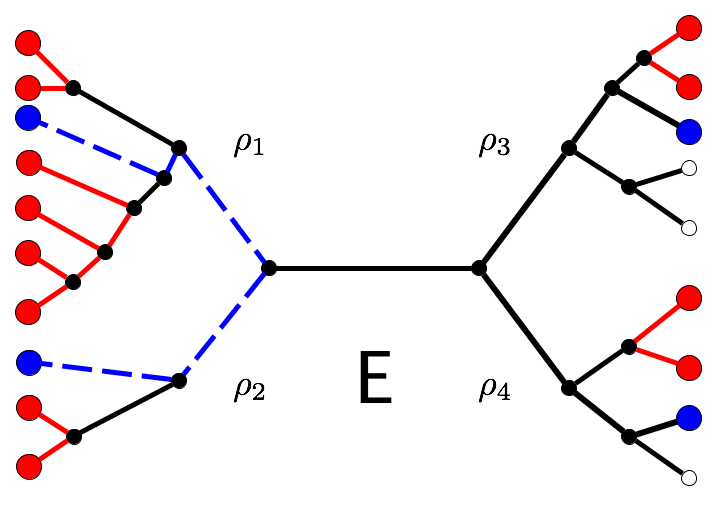}
\end{center}

\end{ex}

To use Draisma's Lemma to prove Theorem 
\ref{thm: main CFN}, we need to construct sets $D_i(v)$ with dimension $2n-3$ for each $1 \leq i \leq r$.
In the following proof, we will show how 
to construct $v$ so that 
 $\C A_i \cap \C S_{2i} \subseteq D_i(v)$.
 By Lemma \ref{lem: CFN slice}, since the affine span of each $\C A_i \cap \C S_{2i}$ has dimension $2n-4$, this ensures
 that dim$(D_i(v)) \geq 2n - 4$. However,
 to prove the theorem, we will also need to ensure each $D_i(v)$  contains a vector outside of the hyperplane $\C S_{2i}$.  For $i = 1$,
this vector will be the vector corresponding to the trivial labeling, which we will call $p_0$.  For each $i > 1$, we will need to ``borrow'' a vector $p_i$ from an adjacent slice.

\begin{proof}[Proof of Theorem \ref{thm: main CFN}]

Note that for $n \geq 7$,
 \[ |\C A_i \cap \C S_{2i}| = \binom{n}{2i} \geq \binom{n}{2} > 2n-3. \]
 Therefore, the set $\C A_i \cap \C S_{2i}$ must have some linear dependencies. 
Choose $p_i \in \C A_i \cap \C S_{2i}$ such that $\langle (\C A_i \cap \C S_{2i}) \setminus \{p_i\}\rangle$ still has dimension $2n-3$.  Let $L_i$ be the set of leaf edges labeled 1 in the labeling corresponding to the vector $p_i$.

For each $1 \leq i \leq r$ let $\pi_i:\B R^{2(2n-3)} \to \B R^{2n}$ be the projection that forgets the coordinates of the non-leaf edges.
Let $\C H_1,\ldots,\C H_{r-1}$ be the hyperplanes in $\B R^{2n}$ with $\C H_i$ defined by
 \[ \sum_{j \in [n]} x_1^j + \frac{2}{4i-1}\sum_{j \in L_i} x^j_1 = 2i + 1. \]
This hyperplane is constructed so that $\pi_i(q) \in \C H_i^-$ for all $q \in (\C A_i \cap \C S_{2i}) \setminus \{p_i\}$ but $\pi_i(p_i) \in \C H_i^+$ and $\pi_{i+1}(q) \in \C H_i^+$ for all $q \in \C A_{i+1} \cap \C S_{2i+2}$.

By Proposition \ref{prop:hyperplanes} there is a sequence of functionals $v' = (v'_1,\ldots,v'_r)$ such that $P_i=R_i(v')\cap P$ (where $P$, $P_i$ and $R_i(v')$ are defined in Proposition \ref{prop:hyperplanes}). 
Letting
 \[ v = (v'_1 \circ \pi_1,\ldots, v'_r \circ \pi_r), \]
we have $(\C A_i \cap \C S_{2i}) \cup \{p_{i-1}\} \setminus \{p_i\} \subseteq D_i(v)$ for $1 \leq i \leq r$ and so the dimension of $D_i(v)$ is $2n-3$.  By Theorem \ref{tropjoin} and the comments after, $V^{\B Z/2\B Z}_{\C T_1} * \cdots * V^{\B Z/2\B Z}_{\C T_r}$ has the expected projective dimension, $r(2n-3)+r-1$.
\end{proof}

\subsection{Group-based models with $|G|>2$}
We continue to assume the trees $\C T_1,\ldots,\C T_r$ are binary and that $B$ is trivial, but consider group $G$ with $|G| > 2$.

\begin{thm}\label{thm: main nonCFN}
Let $\mathcal{T}_1,\ldots, \mathcal{T}_r$ be 
binary phylogenetic $[n]$-trees with $n \geq 2r+4$, and let $G$ be an abelian group with $|G| > 2$.  Then $V^G_{\C T_1} * \cdots * V^G_{\C T_r}$ has the expected projective dimension, $(|G| -1)r(2n-3)+r -1$.
\end{thm}

The proof follows the same structure as Theorem \ref{thm: main CFN}, but requires slightly different arguments.  One way that the case $|G| > 2$ is actually simpler is that for any integer 
$2 \leq j \leq n$ there exists
a constistent leaf-labeling that labels exactly $j$ leaves by non-identity elements of $G$ (this is only the case for even $j$ when $G = \B Z/2\B Z$).  Indeed, we have the following useful fact: Let $G$ be a finite group with order $|G| > 2$.  For any $g \in G$ and any $N \geq 2$, $g$ can be expressed as the sum of exactly $N$ non-identity elements.

We again will partition the vectors
corresponding to constistent leaf labelings based on the number of leaves labeled by non-identity elements. But, we will not need to  ``borrow'' the vectors $p_i$ from adjacent slices as in the proof of Theorem \ref{thm: main CFN}.

While in this section we are focused on binary trees, we prove Lemma \ref{lem: construct v for any g} for the more general, non-binary case.

\begin{lemma} 
\label{lem: construct v for any g}
Let $\mathcal{T}_1,\ldots, \mathcal{T}_r$ be 
phylogenetic (not necessarily binary) $[n]$-trees with $m_1,\ldots,m_r$ edges respectively.  Let
$\C A_i$ be the set of vectors describing the monomial map 
parameterizing $V_{\mathcal{T}_i}^{G}$. 
There exists 
$v = (v_1,\ldots,v_r)$ in 
$\prod_{i=1}^{r} (\mathbb{R}^{|G|m_i})^*$
such that $D_1(v)$ contains the vectors of $\C A_1$ corresponding
to consistent leaf-labelings of $\mathcal{T}_1$ with 
$0$,$2$, or $3$ non-identity labels and
$D_i(v)$ contains the vectors of $\C A_i$ corresponding
to consistent leaf-labelings of $\mathcal{T}_i$ with 
$2i$ or $2i+1$ non-identity labels for $2\leq i \leq r$.
\end{lemma}

\begin{proof} Each tree $\C A_i$ is in $\B R^{|G|m_i}$.  Let $\pi_i: \B R^{|G|m_i} \to \B R^{|G|n}$ be the projection which forgets the coordinates of the non-leaf edges.  Let $\C H_1,\ldots,\C H_{r-1}$ be the parallel hyperplanes in $\B R^{|G|n}$ with $\C H_j$ defined by
 \[ \sum_{i \in [n]} \sum_{g \in G \setminus \{0\}} x_g^i = \frac{3}{2} + 2j. \]
These planes partition $\B R^{|G|n}$ into sets $R_1,\ldots,R_r$ where the projections of labelings with $0,2,$ or $3$ non-identity leaves are in $R_1$ and projections of labelings with $2i$ or $2i + 1$ non-identity leaves are in $R_i$ for $2 \leq i \leq r$.
By Proposition \ref{prop:hyperplanes} there is a sequence of functionals $v' = (v'_1,\ldots,v'_r)$ with $v'_i$ the minimum on $R_i$ and 
 \[ v = (v'_1 \circ \pi_1,\ldots, v'_r \circ \pi_r) \] is the
 desired functional.\end{proof}

\begin{lemma} 
\label{lem: dim D general g}
Let $\mathcal{T}_1,\ldots, \mathcal{T}_r$ be 
phylogenetic $[n]$-trees with $n \geq 2r + 4$.
Let $1 \leq i \leq r$ and let $v$ and $D_i(v)$ be as in Lemma \ref{lem: construct v for any g}.
If $\C T_i$ is binary then $\dim \langle D_i(v) \rangle \geq (|G| -1)(2n-3)+1.$
\end{lemma}

\begin{proof}
It will be convenient to work with the dehomogenized vectors, so let $\pi:\B R^{|G|(2n-3)} \to \B R^{(|G|-1)(2n-3)}$ be the projection which forgets coordinates $x^E_0$ for identity element $0 \in G$ for each edge $E \in \C E(\C T_i)$.  In order to prove the lemma, we will show that for any edge $E \in \C E(\C T_i)$ and any non-identity $g \in G$, that $e^E_g$ is in the span of $\pi(D_i(v))$.

Choose any internal vertex $v$ of $\mathcal{T}_i$ and redraw $\mathcal{T}_i$ as below where each $r_j$ is a rooted subtree of $\mathcal{T}_i$ with root $\rho_j$.

\begin{center}
\includegraphics[width=5cm]{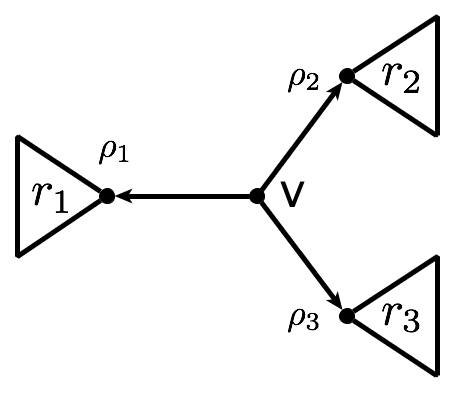}
\end{center}

Choose the orientation of the edges as shown and so that all edges 
in $r_j$ are directed away from $\rho_j$. 
For $1 \leq j \leq 3$,
choose $L_j$ to be a subset of the leaves of 
$r_j$ so that each
$|L_j|$ is odd and so that
$|L_1| + |L_2| + |L_3| =  2i +1$.
This is always possible since $i \leq r$ and $n \geq 2r + 4$. Now we will distinguish a particular leaf in each $r_j$.
There is a unique path from the root $\rho_j$ down $r_j$ such that every vertex on the path has an odd number of leaves in $L_j$ as descendants. 
Label the terminus of this path in $r_j$ by $\lambda_j$. 
Construct a consistent leaf-labeling that labels each of the $\lambda_j$ by the identity and the other $2i - 2$ leaves by non-identity
elements of $G$ so that 
in the induced consistent edge-labeling, every edge on the paths from $\rho_j$ to the $\lambda_j$ is labeled by the identity. 
By our construction, this is always possible using any single non-identity element and its inverse.
Let $F$ be the vector that corresponds to this consistent leaf-labeling.

Now fix non-identity $g \in G$. For any $h \in G$, let $f^j_h$ be the vector with a $1$ in the entry corresponding to $e^E_h$ for each 
edge $E$ on the path from leaf $\lambda_j$ to $v$ and all other entries equal to zero.
Then for any $h \in G$, 
$(f^1_h + f^2_{-h} + F)$ and 
$(f^3_h + f^2_{-h} + F)$ 
are in $\Span{\pi(D_i(v))}$, so their difference, 
$f^1_h - f^3_h$ is in the span. 
Likewise, $f^1_h - f^2_h$  is in the span. 
For any $2 \leq m \leq d-1$
define the vector 
\begin{align*}
v^1_m :=& 
(f^1_{-mg} + f^2_g + f^3_{(m-1)g} + F ) + 
(f^1_g - f^2_g) + 
(f^1_{(m-1)g} - 
f^3_{(m-1)g}) - \\
&(f^1_{mg} +
 f^2_{-mg} + F) - 
 (f^1_{-mg} - 
 f^2_{-mg}) \\
=& (-f^1_{mg} + f^1_g + f^1_{(m-1)g}).
\end{align*}
Then since each $v^1_m \in \Span{ D_i(v)}$, the vector
$$(f^1_{g} +
 f^2_{(d-1)g} + F) +
 (f^1_{(d-1)g} - 
 f^2_{(d-1)g}) + 
 v^1_{d-1} + \ldots + v^1_2 =
 df^1_g + F$$
 is in $\Span{ D_i(v)}$.
Let $a,b,c \in G$ be non-identity elements with $a + b + c = 0$.  Then in $\Span{ D_i(v)}$ is also the vector
\[ (f^1_a + f^2_{-a} + F) + (f^2_b + f^3_{-b} + F) + (f^1_{-c} + f^3_c + F)\]\[ - (f^1_a + f^2_b + f^3_c + F) - (f^1_{-c} + f^2_{-a} + f^3_{-b} + F) \]
\[ = F. \]
Consequently, we have
$f^1_g \in \Span{\pi(D_i(v))}$.
 
To show that $e^E_g$ is in $\Span{\pi(\C A \cap \C H^-)}$ we perform induction on the length of the path represented by $f^E_g$.  For any leaf $E$, we can choose the vertex $v$ incident to this leaf edge so that $f^1_g = e^E_g$.  If the path represented by $f^E_g$ has length $n$, all the edges $E' \neq E$ in the path have their path to a leaf shorter than $n$ so we assume that $e^{E'}_g$ is in the span.  Subtracting these from $f^E_g$ leaves only $e^E_g$ so it is in the span as well.

This proves that $\Span{\pi(D_i(v))}$ has full dimension $(|G|-1)(2n-3)$.  It remains to show that $\Span{D_i(v)}$ has dimension one greater.  In the above argument $df_g^1$ is produced as an integer combination of vectors in $\pi(D_i(v))$.  Note that the sum of the integer coefficients of this linear combination is 0, so there is a vector $w \in \Span{D_i(v)}$ with $\pi(w) = f_g^1$ satisfying $\sum_{g \in G}x_g^E = 0$ for all edges $E$; in particular, $w$ is the vector with a $1$ in the entry corresponding to $e^E_g$, a $-1$ in the entry corresponding to $e^E_0$  for each edge $E$ on the path from $\lambda_1$ to $v$,  and a zero in all other entries. Let $Z$ denote the subspace of $\B R^{|G|(2n-3)}$ defined by the $2n-3$ equations $\sum_{g \in G}x_g^E = 0$, so $Z$ has dimension $(|G|-1)(2n-3)$.  Similarly each vector $e_g^E \in \Span{\pi(D_i(v))}$ can be produced as a linear combination of projections of vectors in $Z$, so $e_g^E - e_0^E \in \Span{D_i(v)}$ for all edges $E$ and all $g \in G \setminus \{0\}$.  Therefore $\Span{D_i(v)}$ contains $Z$.  Note however that any vector in $D_i(v)$ is outside of $Z$, so
 \[ \dim \Span{D_i(v)} \geq (|G|-1)(2n-3) + 1. \]
\end{proof}

Combining Lemma \ref{lem: construct v for any g} and \ref{lem: dim D general g} gives a proof of Theorem \ref{thm: main nonCFN}.\\

\section{Non-binary Trees and other group-based models}
\label{sec: non-binary trees and other group-based models}

In this section, we show that many of the results of the previous sections generalize to non-binary trees and to group-based models where we identify the parameters of some group elements.  First we tackle the case of non-binary trees. 
\subsection{Non-binary trees}
\label{sec: non-binary trees}

To prove the result for non-binary trees, we use the fact that
any $[n]$-tree with no degree-two vertices can be resolved into a binary $[n]$-tree. We then apply our our construction for binary trees to a resolution of each tree and adapt this construction to
obtain an analogous result for the unresolved tree.

\begin{lemma}\label{lem: nonbinary}
Let $\mathcal{T}_1,\ldots, \mathcal{T}_r$ be 
phylogenetic $[n]$-trees with $n \geq 2r + 4$ and let $m_i$ be the number of edges of $\C T_i$ for $1 \leq i \leq r$.
There exist functionals $v = (v_1,\ldots,v_r)$ on the parameter spaces of $V_{\C T_1}^{G},\ldots,V_{\C T_r}^{G}$ such that $\dim \conv(D_i(v)) = (|G|-1)m_i$ for $1 \leq i \leq r$.
\end{lemma}
\begin{proof}
 Let $\C S_i$ be a binary phylogenetic $[n]$-tree obtained by resolving $\C T_i$. That is, $\C T_i$ is obtained after contracting $2n-3-m_i$ internal edges o $\C S_i$.

Again we will work with dehomogenized vectors.  For $1 \leq i \leq r$ let $\pi'_i: \B R^{|G|(2n-3)} \to \B R^{(|G|-1)(2n-3)}$ be the projection which forgets coordinates $x_0^E$ for identity element $0 \in G$ from the parameter space for $V_{\C S_i}^G$.  Let $\pi_i: \B R^{|G|m_i} \to \B R^{(|G|-1)m_i}$ be the equivalent projection for $V_{\C T_i}^G$.

Let $p_i: \B R^{(|G|-1)(2n-3)} \to \B R^{(|G|-1)m_i}$ be the projection from the reduced parameter space of $V_{\C S_i}^G$ to the reduced parameter space of $V_{\C T_i}^G$ which forgets the coordinates of the parameters for the contracted edges.  The kernel of $p_i$ has dimension $(|G|-1)(2n-3-m_i)$.  If $\C A'_i$ is the set of exponent vectors associated to the toric variety $V_{\C S_i}^G$ and $\C A_i$ the set of vectors corresponding to $V_{\C T_i}^G$, then $\pi_i(\C A_i) = p_i (\pi'_i(\C A'_i))$.

Choose functionals $v' = (v'_1,\ldots,v'_r) \in \prod_{i=1}^r (\mathbb{R}^{|G|(2n-3)})^*$ 
as in the proof of Theorem \ref{thm: main CFN} or Theorem \ref{thm: main nonCFN} (depending on whether $G = \B Z/2\B Z$).  As shown above, $\dim \conv(D_i(v')) = (|G|-1)(2n-3)$. Additionally $\dim \conv(\pi'_i(D_i(v'))) = (|G|-1)(2n-3)$ since $\pi'_i$ does not affect the dimension of subsets of $\C A'_i$. Let $v=(v_1,...,v_r)$ be in $\prod_{i=1}^r (\mathbb{R}^{|G|m_i})^*$ where $v_i$ is $v'_i$ after projecting away the $|G|(2n - 3 - m_i)$ entries corresponding to the contracted edges.  Each $v'_i$ depends only on the coordinates of the leaf edges of $\C S_i$, so each consistent labeleing has the same evaluation by $v'$ and $v$.  The minimum value is achieved at the same index $i$ so
 \[ \pi_i(D_i(v)) = p_i(\pi'_i(D_i(v'))). \]
 
We can conclude
\[\dim \conv(D_i(v)) \geq \dim \conv(\pi_i(D_i(v))) \]\[\geq \dim \conv(\pi'_i(D_i(v'))) - \dim \ker(p_i) = (|G|-1)m_i. \]
However $\dim \conv(D_i(v'))$ cannot exceed $(|G|-1)m_i$ because $D_i(v) \subseteq \C A_i$ and the affine span of $\C A_i$ has dimension $(|G|-1)m_i$.
\end{proof}
The above result allows us to extend Theorem \ref{thm: main nonCFN} to non-binary trees, and henceforth we can work more generally with arbitrary trees with no degree-2 vertices.

\subsection{Identifying Parameters}
\label{sec: Identifying Parameters}

The final generalization is to allow parameters to be identified according to a non-trivial subgroup $B$ of $\Aut(G)$. Recall that all group elements in the same orbit of $B$ are assigned the same parameter. As discussed in the introduction, many of the most commonly used models in phylogenetics, including the JC and K2P models, are models of this form.

\begin{lemma}\label{lem: nonbinary eqclass}
Let $\mathcal{T}_1,\ldots, \mathcal{T}_{r}$ be 
phylogenetic $[n]$-trees with $n \geq 2r + 4$ and let $m_i$ be the number of edges of $\C T_i$ for $1 \leq i \leq r$.  Let $G$ be a finite abelian group with $B$ a subgroup of $\Aut(G)$.
There exist functionals $v = (v_{1},\ldots,v_{r})$ on the parameter spaces of $V_{\C T_{1}}^{(G,B)},\ldots,V_{\C T_{r}}^{(G,B)}$ such that $\dim \conv(D_{i}(v)) = lm_{i}$ for $1 \leq i \leq r$.
\end{lemma}
\begin{proof}
Lemma \ref{lem: nonbinary} proves the case that $B$ is trivial.  Assume that $B$ is non-trivial in which case $|G| > 2$.

For $1 \leq i \leq r$ let $p_i: \B R^{|G|m_{i}} \to \B R^{(l+1)m_i}$ be the projection from the parameter space of $V_{\C T_i}^{(G,1)}$ to the parameter space of $V_{\C T_i}^{(G,B)}$ by summing the coordinates of the parameters that are identified in $B$ for each edge.  The kernel of $p_i$ has dimension $(|G|-l-1)m_i$.  If $\C A'_i$ is set of exponent vectors associated to toric variety $V_{\C T_i}^{(G,1)}$ and $\C A_i$ the vectors to $V_{\C T_i}^{(G,B)}$ then $\C A_i = p_i(\C A'_i)$.

Choose functionals $v' = (v'_1,\ldots,v'_r)$ with $v'_i$ acting on $\B R^{|G|m_i}$ as in the proof of Lemma \ref{lem: nonbinary}.  As was shown above, $D_i(v')$ has affine dimension $(|G|-1)m_i$.  Because we are in the case $|G| > 2$, each $v'_i$ depends only on the total number of non-identity leaf labels, which does not change when parameters are identified according to $B$.  Therefore there is a functional $v_i$ on $\B R^{(l+1)m_i}$ such that $v'_i = v_i \circ p_i$.  Let $v = (v_1,\ldots,v_r)$.  Each constistent labeleing has the same evaluation by $v'$ and $v$ so the minimum value is achieved at the same index $i$.  Then
 \[ D_i(v) = p_i(D_i(v')). \]
The dimension of $D_i(v)$ has the bound
 \[ \dim \conv(D_i(v)) \geq \dim \conv(D_i(v')) - \dim \ker(p_i) = lm_i. \]
Since $lm_i$ is the dimension of the affine span of $\C A_i$, it is also an upper bound on $\dim \conv(D_i(v))$.
\end{proof}

Applying Draisma's Lemma to Lemma \ref{lem: nonbinary eqclass} shows that for $\C T_1,\ldots,\C T_r$ phylogenetic $[n]$-trees (not necessarily binary) with $n \geq 2r+5$, $G$ an abelian group, and $B$ a subgroup of $\Aut(G)$, the mixture model $V_{\C T_1}^{(G,B)}*\cdots *V_{\C T_r}^{(G,B)}$ has projective dimension $lM + r - 1$ where $M$ is the sum of the number of edges among $\C T_1,\ldots,\C T_r$ and $l+1$ is the number of orbits of $B$ in $G$.  This is the expected dimension, so we have completed the proof of Theorem \ref{thm: main} in full generality.

\section{Improved bounds for special cases}.
\label{sec: special cases}
Our proof of Theorem \ref{thm: main} holds when the number $r$ of phylogenetic tree models in the mixture is not too large compared to number $n$ of leaves of the trees, according to bound $n \geq 2r + 5$.  It should be noted though that the bound $n \geq 2r + 5$ merely reflects the limitations in our proof techniques. In our experiments we have not come across any defective mixtures of phylogenetic tree models, and we have no reason to believe that these models have defective join dimensions for larger $r$, so we state the following conjecture.

\begin{conj}\label{conj: main}
Let $\mathcal{T}_1,\ldots, \mathcal{T}_r$ be 
phylogenetic $[n]$-trees with $n \geq 3$, and let $G$ be an abelian group and $B\subset \Aut(G)$.  Then $V^{(G,B)}_{\C T_1} * \cdots * V^{(G,B)}_{\C T_r}$ has the expected dimension.
\end{conj}

\subsection{Claw trees}

In some special cases we can improve the bound.  For instance, Theorem \ref{thm: main nonCFN} states that when the group $G$ has order at least 3, then joins have the expected dimension for $n \geq 2r+4$. When each tree in the mixture is the $n$-leaf claw tree, we can improve this bound.

\begin{prop}\label{prop:claw}
 Let $\mathcal{T}$ be the $[n]$-leaf claw tree and let $G$ be an abelian group and $B\subset \Aut(G)$.  Then for $n \geq 2r + 1$, the $r$th secant $\sigma_r(V^{(G,B)}_{\C T})$ has the expected dimension.
\end{prop}

\begin{proof}The proof outline follows that of Theorems \ref{thm: main CFN} and \ref{thm: main nonCFN}, but with simplifications that allow for the improved bound $n \geq 2r+1$.

For $G = \B Z/2\B Z$, as in Lemma \ref{lem: CFN slice} let $\C S_c$ the hyperplane defined by
\[ \sum_{E \in \C L(\C T)} x^E_1 = c. \]
We show that for even $2 \leq c \leq n-1$, $\rank\Span{\C A \cap \C S_c}= n$.  For any pair of edges $E_1,E_2$, let $L$ be any collection of $c-1$ edges not containing $E_1$ or $E_2$ and $F = \sum_{j\in L} e^j_1$.  Then in $\Span{\C A \cap \C S_c}$ is
 \[ (e^{E_1}_1 + F) - (e^{E_2}_1 + F) = e^{E_1}_1 - e^{E_2}_1. \]
Fixing edge $E$, now let $L$ be any set of $c-1$ edges not containing $E$ and again $F = \sum_{j\in L} e^j_1$.  The vector
 \[ (e^E_1 + F) + \sum_{j \in L} (e^E_1 - e^j_1) = ce^E_1 \]
is in $\Span{\C A \cap \C S_c}$ for every edge $E$ so $\rank\Span{\C A \cap \C S_c}=n$.

The remainder of the proof exactly follows the proof of Theorem \ref{thm: main CFN} replacing dimension $2n-3$ (the number of edges of a binary $[n]$-tree) with $n$ (the number of edges of the claw $[n]$-tree).

For $|G| > 2$, as in Lemma \ref{lem: construct v for any g} choose functionals $v = (v_1,\ldots,v_r)$ that divide the vectors corresponding to consistent leaf-labelings so that $D_i(v)$ contains the leaf-labelings with $2i$ or $2i+1$ non-identity edges for $1 \leq i \leq r$.

Let $\pi: \B R^{|G|n} \to \B R^{(|G|-1)n}$ be the dehomogenization map that forgets the coordinates of $x_0^E$ for each edge $E$.  Working in the dehomogenized coordinates, fix $i$ and any edge $E_1$.  Let $L$ be any collection of $2i-2$ edges not containing $E_1$ and $F = \sum_{j\in L} e^j_1$.  Choose $E_2$ and $E_3$ to be additional edges not in $L$.  Follow the argument in the proof of Lemma \ref{lem: dim D general g} but replacing each $f^j_h$ with $e^{E_j}_h$.  This shows that $e^{E_1}_g \in \Span{\pi(D_i(v))}$ for any non-identity $g \in G$, and consequently that
 \[ \dim \Span{D_i(v)} = (|G|-1)n + 1. \]
Applying Draisma's Lemma, this proves the result for $B$ trivial.  The argument in the proof of Lemma \ref{lem: nonbinary eqclass} can be applied here for the case that $B$ is non-trivial.
\end{proof}

\subsection{Trees with few leaves}

For a specific value of $r$ and a specific model $M = (G,B)$, there are a finite number of collections of $[n]$-trees $\C T_1,\ldots,\C T_r$ with $n < 2r + 5$, so one can check whether all joins have the expected dimension by explicit computation.  We perform some of these computations in the computer algebra system {\tt Macaulay2} \cite{grayson2002macaulay} using the package {\tt PhylogeneticTrees} \cite{banos2016phylogenetic}.

One can efficiently compute the dimensions of joins of parametrized varieties using the principle of Terracini's Lemma.  The dimension of a join is equal to the dimension of the tangent space at a generic point on the join variety.  Choosing random parameter values, we obtain the tangent space dimension from the rank of the Jacobian of the paramtrization map.  To further improve efficiency, we compute the rank over a finite field $\B F_p$ for a large prime $p$.

Note that this algorithm is probabilistic.  With small probability the random parameter values may be non-generic, leading to a drop in the dimension of the tangent space.  Additionally, for some parameter values there may be more linear dependencies in the Jacobian over $\B F_p$ than over characteristic zero.  Both situations produce a lower value than the true dimension, so this algorithm only certifies a lower bound.  However, if the algorithm returns a value equal to the expected dimension, it is a proof that the join is not defective.  

Here we state the results that we are able to obtain by combining Theorem \ref{thm: main} and computational results obtained for small $n$.

\begin{prop}\label{prop:computational1}
Let $\mathcal{T}$ be a
phylogenetic $[n]$-tree with $n \geq 3$.  Then the second secant $V^{M}_{\C T} * V^{M}_{\C T}$ has the expected dimension for $M$ equal to $\CFN$, $\JC$, $\KwP$ or $\KtP$.
\end{prop}
This statement was previously proved for models $\JC$ and $\KwP$ in \cite{Allman}. The case of $n \geq 9$ in Proposition \ref{prop:computational1} is proved by Theorem \ref{thm: main}, and $n \geq 8$ for $M \neq \CFN$ by Theorem \ref{thm: main nonCFN}.  For each model $M$ we check the dimension of $V^{M}_{\C T} * V^{M}_{\C T}$ for all $[n]$-trees $\C T$ with $n \leq 7$ and $V^{\B Z/2\B Z}_{\C T} * V^{\B Z/2\B Z}_{\C T}$ for $[8]$-trees $\C T$.  The dimension of the secant is invariant under relabeling the leaves of $\C T$, so we need only check one tree for each equivalence class.  The equivalence classes correspond to the set of unlabeled trees with $n$ leaves.

\begin{prop}\label{prop:computational2}
Let $\mathcal{T}$ be a
phylogenetic $[n]$-tree with $n \geq 3$.  Then third secant $V^{M}_{\C T} * V^{M}_{\C T} * V^{M}_{\C T}$ has the expected dimension for $M$ equal to $\CFN$, and for $M$ equal to $\JC$, $\KwP$ or $\KtP$ with $n \neq 9$.
\end{prop}

For Proposition \ref{prop:computational2}, the case of $n \geq 11$ is proved by Theorem \ref{thm: main}, and $n \geq 10$ for $M \neq \CFN$ by Theorem \ref{thm: main nonCFN}.  For $M = \CFN$ we check the dimension for all $[n]$-trees $\C T$ with $n \leq 10$.  For $M$ equal to $\JC$, $\KwP$ or $\KtP$ we were able to check the dimension for all $[n]$-trees $\C T$ with $n \leq 8$, but were not able to complete the computation on $[9]$-trees.

\begin{prop}\label{prop:computational3}
Let $\mathcal{T}_1,\mathcal{T}_2$ be 
phylogenetic $[n]$-trees with $n \geq 3$.  Then $V^{M}_{\C T_1} * V^{M}_{\C T_2}$ has the expected dimension for $M$ equal to $\CFN$, and for $M$ equal to $\JC$, $\KwP$ or $\KtP$ with $n \neq 7$.
\end{prop}

For Proposition \ref{prop:computational3}, the case of $n \geq 9$ is proved by Theorem \ref{thm: main}, and $n \geq 8$ for $M \neq \CFN$ by Theorem \ref{thm: main nonCFN}.  For $M = \CFN$ we check the dimension of $V^{M}_{\C T_1} * V^{M}_{\C T_2}$ for all pairs $(\C T_1,\C T_2)$ of $[n]$-trees with $n \leq 8$.  Again, the dimension of the join is invariant under permutations of $[n]$, but a relabeling applies to both trees in the pair $(\C T_1,\C T_2)$.  Therefore we can let $\C T_1$ vary over the set of unlabeled trees, but $\C T_2$ must then be chosen from the full set of labeled $[n]$-trees.  For $M$ equal to $\JC$, $\KwP$ or $\KtP$ we were able to check the dimension for all pairs of $[n]$-trees with $n \leq 6$, but were not able to complete the computation on pairs of $[7]$-trees.

\begin{rmk} Note that while we resolve most of the cases for binary trees with few leaves in this section for 3-tree secants and 2-tree joins with respect to the $\JC$, $\KwP$, and $\KtP$ models, the $n=8$ case for 3-tree secants and the $n=7$ case for 2-tree joins remains an open computational question.
\end{rmk}

\section{Discussion}

In this paper, we show that, for $n \geq 2r+5$, the join varieties associated to a large class of group-based models, including the CFN, JC, K2P, and K3P models have the expected dimension.  In order to provide a complete answer to Conjecture \ref{conj: main}, we expect different proof techniques would need to be used to handle the $n < 2r+5$ case.  However, we showed how this bound could be improved for the case of claw trees in the proof of Proposition \ref{prop:claw}.

Not only do mixtures of group-based models give rise to a class of join varieties of toric varieties that are interesting to study in their own right, but the the dimension results in this paper have important statistical applications.  In particular, the dimension of these varieties play a key role in establishing identifiabiilty \cite{Allman, long2017identifiability}, which we discussed in the Introduction, but now define formally here.

\begin{defn} The tree parameters of the $r$-tree mixture model are \emph{generically identifiable} if, for any binary trees $\mathcal T_1$, \ldots, $\mathcal T_r$ on the same set of taxa, and generic choices of $\theta_1$, \ldots, $\theta_r$, $\omega$, the equality
$\psi_{\mathcal T_1 , \ldots, \mathcal T_r} ( \theta_1 , \ldots, \theta_r , \omega ) = \psi_{\mathcal T_1' , \ldots, \mathcal T_r'} ( \theta_1' , \ldots, \theta_r' , \omega' )$
implies $\{\mathcal T_1, \ldots, \mathcal T_r\} = \{\mathcal T_1', \ldots, \mathcal T_r'\}$.  
\end{defn}

In \cite{Allman}, the authors show that the tree parameters of Jukes-Cantor and Kimura 2-parameter 2-tree mixtures are generically identifiable
using phylogenetic invariants and knowledge of the dimension of the join varieties.  However, the question of whether the tree parameters of Kimura 3-parameter 2-tree mixtures are generically identifiable remains open. Likewise, identifiability for mixtures with more than 2 trees remains open. Resolving the Conjecture \ref{conj: main} would be an important step towards these results.

\section{Acknowledgements} This work began at the 2016 AMS Mathematics Research Community on ``Algebraic Statistics," which was supported by the National Science Foundation under grant number DMS-1321794. RD was supported by NSF DMS-1401591. EG was supported by NSF DMS-1620109.
RW was supported by a NSF GRF under grant number PGF-031543,   NSF RTG grant 0943832, and a Ford Foundation Dissertation Fellowship. HB was supported in part by a research assistantship, funded by
the National Institutes of Health grant R01 GM117590. PEH
was partially supported by NSF grant DMS-1620202.

\newpage

\bibliography{references}
\bibliographystyle{plain}

\end{document}